\documentclass[12pt]{article}
\usepackage{amsmath,amssymb,amsthm,braket}
\usepackage{tikz}
\usepackage{caption}
\newtheorem{lemma}{Lemma}
\newtheorem{theorem}{Theorem}

\newtheorem{corollary}{Corollary}
\newtheorem{remark}{Remark}
\newcommand{\Tr}{\operatorname{Tr}}

\title{A Stochastic Reconstruction of Entangled States: from Quantum Foundations to Machine Learning and Decision Making}  
\author{Andrei Khrennikov$^{1,*}$  and Irina Basieva$^{2}$}
\date{}

\begin{document}

\maketitle

\noindent
$^{1,*}$Center for Mathematical Modeling in Physics and Cognitive Sciences\\
Linnaeus University, V\"axj\"o, SE-351 95, Sweden\\
$^{2}$ Quantum Group, School of Computing, Newcastle University\\ 
1 Science Square, Newcastle upon Tyne, NE4 5TG, UK\\
*Corresponding author email: Andrei.Khrennikov@lnu.se

\begin{abstract}
This article presents a concrete mathematical framework for the generation of entangled quantum states from classical stochastic processes. We demonstrate that any density operator $\rho_{AB}$ of a composite system can be derived from the correlations between two underlying stochastic processes, $X(t)$ and $Y(t)$, representing the random fluctuations of its subsystems. This construction utilizes a two-scale temporal scheme—micro and macro time—where quantum correlations emerge as macro-correlations derived from underlying micro-correlations. We propose the Double Covariance Model (DCM), which reproduces the {\bf some properties} of quantum theory by treating quantum states as the fourth-order moment structure of an underlying classical probability space. DCM may also have potential implications for applications, particularly in quantum-inspired simulation and technology, and quantum-like modeling of cognition and decision making, as a tool for generation of density matrices with classical stochastic processes
\end{abstract}

{\bf keywords:} quantum vs classical, generation of entangled states, classical stochastic processes, interplay of macro-micro time scales, composite vs. individual quantum systems, relational quantum mechanics

\section{Introduction}

The inquiry into the relationship between classical probability and the quantum formalism began at the inception of quantum theory over a century ago. This paper presents the {\it Double Covariance Model (DCM)}, a framework that provides a stochastic reconstruction of entangled quantum states through the interplay of micro and macro time scales \cite{ARX}. 

The DCM treats the quantum density operator as a hierarchical statistical construct. It posits that a density operator acting in a tensor product $H_A\otimes H_B$ is effectively a ``covariance of covariances'':
\begin{itemize}
    \item {\it Primary (Micro) Covariance}: The temporal synchronization of subquantum processes $X(t)$ and $Y(t)$ defines a random operator at the micro-scale (which can be identified with a random vector in $H_A\otimes H_B$).
    \item {\it Secondary (Macro) Covariance}: The ensemble stability of these random operators (vectors in $H_A\otimes H_B)$  across a macro-scale temporal window defines the density operator $\rho_{AB}$.
\end{itemize}

From the perspective of multivariate statistics, the DCM interprets the quantum state as the fourth-order moment structure of an underlying classical probability space. This approach demonstrates that any density operator of a composite system can be derived from the interplay between micro-scale cross-covariance and macro-scale covariance - by employing {\it the fourth order moment structure.}

We point out that, despite its primary orientation toward quantum foundations, the proposed framework may also have potential implications for applications, particularly in quantum-inspired simulation and technology {\bf and quantum-like modeling for cognition and decision making. The DCM provides a basis for the numerical and analogue generation of quantum entangled states using classical stochastic processes.}

\subsubsection*{Entanglement as micro-time consistency}
A central thesis of DCM is that entanglement {\bf (as a mathematical construction)} is not based on the statistical dependence of subquantum processes at the macro level. Instead, entanglement is interpreted as a macro-time phenomenon reflecting {\it micro-time consistency}. In this framework, subquantum processes can be statistically independent while remaining pathwise constrained (consistent) at the micro-scale. For example, Bell states can be generated by partitioning a macro-time window into subintervals where specific micro-correlations are satisfied - {\bf we call this approach the {\it subinterval allocation scheme.}} 

\subsubsection*{Foundations and context}
The DCM is part of a long lineage of attempts  {\bf (having various degrees of success)} to bridge the classical and quantum regimes, including
Wigner functions - providing a quasi-probability distribution in phase space \cite{Wig32, Hil84}, De Broglie's Double Solution Model - An early attempt at a causal, wave-particle duality \cite{DeB27, DeB60}, von Neumann's No-Go Theorem - the early formal argument against hidden variables \cite{VN}, Bell's Theorem - establishing the boundaries of local realism \cite{Bel64, Asp82, Bel66}, Stochastic Electrodynamics (SED) - attributing quantum effects to a classical zero-point field \cite{AM1,AM2,Col77,AM3}, Hydrodynamic Models - representing the Schr\"{o}dinger equation as fluid dynamics \cite{Mad27, Wya05},
Bohmian Mechanics - a deterministic, non-local pilot-wave theory \cite{Boh52,Dur09},
Kochen-Specker Theorem - Highlighting the role of contextuality \cite{KS67, Mer90}, 
Hydrodynamic Droplet Systems - classical fluid dynamics with pilot-wave behavior \cite{Bush1,Bush2,Bush3},
Quantum–classical Hybrid Models - coupling quantum and classical variables within a single framework 
\cite{Elze1,Elze2}, Prequantum Classical Statistical Field Theory (PCSFT) - theory of prequantum random fields 
\cite{PCSFT1,PCSFT2,PCSFT3,Beyond}
 
While PCSFT served as a catalyst for this work, the DCM takes a significant conceptual step forward by grounding entanglement in the dynamic interplay between micro-scale cross-covariance and macro-scale covariance {\bf and moving from second order to fourth order moment structure.} 

{\bf It is important to emphasize that DCM is still at an early stage of development and that its current results should not be overstated as achieving a reduction of quantum theory to classical stochastic processes. Nevertheless, the possibility of generating entangled states and arbitrary density operators represents an important step forward. In a preprint \cite{ArxL}, we demonstrated that DCM leads to the Gorini–Kossakowski–Sudarshan–Lindblad (GKSL) equation and, as a particular case, to the von Neumann equation. This derivation relies on the interplay between microscopic and macroscopic time scales, closely connected to the framework of hydrodynamic limits.

In the present paper, we investigate another important aspect of coupling between DCM and quantum theory: the relationship between the states of a composite system 
$S_{AB}$ and those of its subsystems $S_A$ and $S_B.$ This problem turns out to be far from straightforward. In particular, the validity of the partial trace formula depends on the class of classical stochastic processes underlying the density operators. We demonstrate that the formula holds as an exact equality for stochastic processes constructed within our subinterval allocation scheme. We also consider a different class of stochastic processes for which the partial trace formula holds only approximately. Since the mathematical construction of this second class is considerably more involved, we relegate its treatment to appendix A.

A further step in applying DCM to quantum foundations is the development of a corresponding measurement theory. This is a challenging project, both mathematically and conceptually. Our tentative plan is to follow an approach analogous to that of PCSFT, based on quadratic forms of classical stochastic processes. At this stage, however, it would be premature to speculate on whether this strategy will ultimately yield a successful measurement theory.}

\subsubsection*{Models in physics based on micro-macro time scale interplay}

In various branches of physics, the transition from fundamental fluctuations to observable macroscopic behavior is often modeled through the interplay of two or more distinct time scales. This section summarizes key models where a micro scale (rapid fluctuations or collisions) and a macro scale (ensemble stability or order parameters) are utilized.

{\it Langevin dynamics and Brownian motion \cite{LA1,LA2}: }
The classic description of a macroscopic particle suspended in a fluid. The micro-scale consists of rapid, stochastic collisions with fluid molecules ($10^{-12}$ s), while the macro-scale describes the observable diffusion of the particle. The interplay is captured by the Langevin equation, where micro-fluctuations are modeled as a white noise term.

{\it Haken's synergetics and the slaving principle  \cite{Haken1975,Haken1983}:}
This framework describes self-organizing systems. It relies on the adiabatic elimination of fast-relaxing micro variables. The macro behavior is governed by a few slow-moving order parameters that ``slave'' the micro-components, leading to emergent patterns in lasers and fluids.

{\it Kinetic theory and the BBGKY hierarchy  \cite{Kirkwood,Grad1949}:}
In statistical mechanics, the transition from reversible micro-dynamics to irreversible macro-thermodynamics requires a hierarchy of scales. The micro-scale is the collision time, while the macro-scale is the relaxation time to equilibrium. The Boltzmann equation emerges by coarse-graining the micro-correlations \cite{Kirkwood,Grad1949}.

{\it Stochastic electrodynamics \cite{AM1}-\cite{AM3}:}
SED posits that quantum effects arise from the interaction of classical particles with a classical, stochastic zero-point field (ZPF). The micro-scale involves the high-frequency fluctuations of the ZPF, while the macro-scale involves the averaged motion of particles that mimics quantum mechanics.

{\it Brownian entanglement \cite{ALA}:} For two interacting classical Brownian particles, the separation of micro and macro time scales generates coarse-grained velocity–position correlations that cannot be factorized, creating a classical analog of entanglement. This micro–macro time interplay produces entanglement-like correlations that vanish under finer temporal resolution.

\subsubsection*{The relational nature of systems}
The DCM further challenges the notion of the isolated system. It suggests that the distinction between ``composite'' and ``individual'' systems is relational rather than ontological. In this view:
\begin{itemize}
    \item A \textbf{composite system} is one where micro-synchronization between internal processes is explicitly resolved.
    \item An \textbf{individual system} is a marginal residue of a larger, synchronized whole, where internal synchronization is treated as a unified, emergent fluctuation.
\end{itemize}

The paper proves that the partial trace operation in quantum mechanics is equivalent to the marginalization of hidden correlations in the underlying classical space. Thus, the ``Quantum State'' of an individual system is not a standalone primitive but a reduced description of its participation in a larger global field.

\subsubsection*{Coupling with Relational Quantum Mechanics}

The DCM motivates a relational interpretation 
\cite{R1}-\cite{R3} of quantum states, in which the distinction between composite and individual systems depends on the scale at which their internal correlations are resolved. This perspective suggests possible connections with relational quantum mechanics. The precise physical interpretation of this correspondence, however, remains an open question.

\subsubsection*{Applications to quantum-inspired machine learning}

{\bf In recent years, hybrid quantum-inspired algorithms have attracted increasing interest across a broad range of fields, from biology and signal processing to engineering and data analysis \cite{ML}. These approaches combine classical data processing with computational procedures inspired by quantum theory, while being implemented entirely on classical digital computers. A typical strategy is to encode classical data, such as signals from multiple sensors, into a quantum state, process this state using a quantum-inspired algorithm, and extract a set of probabilities through the Born rule. These probabilities are then supplied to a classical machine-learning model, for example, a neural network, which produces the final output, such as a binary classification or a decision.

The classical-to-quantum encoding procedure plays a crucial role in this framework. Standard encoding schemes often produce separable states of composite systems, thereby restricting the correlations represented by the resulting quantum state. The DCM framework suggests a potentially different approach. By applying DCM directly to classical input data, one may construct density operators that are not restricted to separable states and can, in particular, represent entangled states. Such states could then serve as inputs to quantum-inspired algorithms implemented on classical computers. This approach would provide a systematic way to explore whether the richer correlation structures associated with entangled-state representations can improve the performance of machine-learning models.
}

\subsubsection*{Quantum-like modeling in cognition and decision making}

{\bf Quantum-like cognition and decision making (QCDM) is a rapidly developing research area attracting interest from researchers in cognitive science, psychology, economics, finance, and decision theory \cite{QL1,QL2,QL3}. It is important to distinguish QCDM sharply from attempts by Hameroff and Penrose to associate mental processes with quantum physical processes in the brain, particularly those hypothesized to occur in microtubules. Rather than proposing a quantum physical mechanism underlying cognition, QCDM employs quantum probability theory as a phenomenological framework for modeling contextuality, interference, and order effects in cognitive processes and social information flows.

A fundamental open question concerns how QCDM relates to the functioning of biological neural networks and social networks. One possible approach is to construct quantum-like representations of classical statistical data generated by these networks. In particular, can entangled mental states be constructed from signals produced by biological neural networks, and what would such states represent in terms of cognitive processes? Initial attempts to establish this neuro-quantum connection have 
been discussed in \cite{Khr10}. We expect that DCM could provide a useful mathematical framework for investigating these questions by constructing quantum-state representations from classical stochastic processes. Whether this approach can yield meaningful models of cognition and decision making remains an open problem.}

\section{The density operator as a double covariance: from micro-scale cross-covariance to macro-scale covariance}

Let $(\Omega, {\cal F}, P)$ be a classical probability space (Kolmogorov \cite{K}): $\Omega$ is  a set of chance parameters (``elementary events''), ${\cal F}$ is collection of events, and $P$ is a probability measure defined on ${\cal F}.$ 

Everywhere in this paper $H_A$ and $H_B$ are finite-dimensional complex Hilbert spaces. Consider two stochastic processes:
\[
X(t): \Omega \to H_A, \quad Y(t): \Omega \to H_B.
\]
where the processes have zero mean value, 
\begin{equation}
\label{Center}
\mathbb{E}[X(t)]=0, \mathbb{E}[Y(t)]=0
\end{equation}
 for any $t\geq 0.$ We also assume that these processes have finite second order moments: 
\begin{equation}
\label{2d}
\mathbb{E}[||X(t)||^2 ]< \infty, \mathbb{E}[||Y(t)||^2] < \infty.  
\end{equation}
They describe stochasticity in two systems $S_A$ and $S_B;$ stochasticity in a composite system 
$S_{AB}=(S_A,S_B)$ is described by the process valued in $H_A \times H_B$ with the coordinate processes $X(t), Y(t).$  Stochasticity under consideration is classical. However,  we will see that it can be represented in quantum-like  way - by a density 
operator. Classicality is a feature on the micro-time dynamics. Transition from micro-time scale to macro-time scale leads to the quantum representation.

So, we consider two time scales: a micro-time scale and a macro-time scale. The micro- and macro-time variables are denoted as $t$ and $\tau.$ The scale of macro-time is determined  by an interval $\Delta,$ this is an instant of macro-time $\tau.$ 
The chance parameter $\omega \in \Omega$ describes en ensemble of intervals $\Delta,$ a sample of instances of macro-time.  

To proceed further at the mathematical level of rigorousness (see appendix B), we complete condition (\ref{2d}) with the condition:  
\begin{equation}
\label{eq:fourth-integrability1}
\int_{\Delta}
\mathbb{E}\!\left[\|X(t)\|^2\,\|Y(t)\|^2\right] \, dt < \infty,
\end{equation}

\paragraph{1. The Micro-scale Cross-Covariance Operator}

 We start with a random vector $|C_{\Delta}\rangle=|C_{\Delta}\rangle (\omega)$ valued in the tensor product $H_A \otimes H_B$:
\begin{equation}
\label{L4}
|C_{\Delta}\rangle = \frac{1}{|\Delta|} \int_{\Delta} |X(t) \rangle \otimes |Y(t)\rangle \, dt
\end{equation}
We remark that it can be represented as a random operator $\hat C_{\Delta}$ valued in the space of linear operators acting from
$H_B$ to $H_A$ (denoted as $L(H_B, H_A)):$
\begin{equation}
\label{L477}
\hat C_{\Delta} = \frac{1}{|\Delta|} \int_{\Delta} |X(t) \rangle  \langle \bar{Y}(t)| \, dt.
\end{equation}
In this representation $\hat C_{\Delta}$ can be interpreted as the temporal covariance of two random processes.\footnote{Here we need the complex conjugation since we cannot simply identify the vector space $H_A\otimes H_B$ with the space of operators 
$L(H_B, H_A).$ The canonical isomorphism is \(H_A\otimes\overline{H_B}\cong\mathcal L(H_B,H_A)\), rather than the identification of \(H_A\otimes H_B\) with \(\mathcal L(H_B,H_A)\).  In our formalism \(x\otimes y\) is represented by the rank-one map 
\(z\mapsto x\langle y,z\rangle\).
}
  
The centered random variable representing the micro-scale fluctuations is 
\begin{equation}
\label{L5}
|Z_{\Delta}\rangle = |C_{\Delta}\rangle - \mathbb{E}[|C_{\Delta}\rangle],
\end{equation}
where $\mathbb{E}$ denotes the mathematical expectation w.r.t. probability $P$ - statistical expectation.
 
\paragraph{3. The Macro-Covariance Operator}
The macro-covariance operator $\hat{C}$ is defined as the covariance of the $H_A\otimes H_B$-valued random variable $|Z_{\Delta}\rangle$. Following the convention of linearity in the second argument:
\begin{equation}
\label{L6}
\hat{C} = \mathbb{E}\left[ |Z_{\Delta}\rangle \langle Z_{\Delta}| \right] \in {\cal L}(H_A \otimes H_B)
\end{equation}
This operator $\hat{C}$ is Hermitian and positive semi-definite.\footnote{If, instead of the vector representation, 
 we use the operator representation $\hat C_\Delta,$ then  $\hat{C}$ can be interpreted as statistical 
covariance of temporal covariances, so to say {\it double covariance.} This quantity is determined by the fourth order 
moment structure of the classical stochastic processes.}   

DCM generates a density operator via normalization
 \begin{equation}
\label{DO}
\rho = \hat{C}/\text{Tr} \hat{C},
\end{equation}
whenever $\text{Tr}\hat{C} >0.$

\subsection{Density operator from micro-scale time series}

This abstract framework can be operationalized through the following scheme, which allows for experimental verification.

We consider again two time scales: a micro-time scale and a macro-time scale. The micro-time variable are denoted by $t.$
The scale of macro-time is determined  by an interval $\Delta,$ this is an instant of macro time.   
The macro-time variable is denoted by  $\tau;$ in the discrete framework:$
\tau_k = k\Delta, \qquad k = 0,1,2,\dots .$ We define the associated micro-scale time windows as
$
\Delta_k := [k\Delta,(k+1)\Delta).
$
Fix a sufficiently large integer $N$ and a macro-scale time interval
$
[0,\mathcal{T}], \qquad \mathcal{T} = N\Delta,
$
so that $
[0,\mathcal{T}) = \Delta_0 \cup \cdots \cup \Delta_{N-1},
$
where the intervals $\Delta_k$ are non-overlapping.

Consider two time series $X(t_i)$ and $Y(t_i)$, with $i=1,2,\dots$. For each interval $\Delta_j$ (corresponding to the macro-time instant $\tau_j$), we compute the micro-scale cross-moment
\begin{equation}
\label{L7}
C(\tau_j) :=
\frac{1}{| \{t_k \in \Delta_j\}|}
\sum_{t_k \in \Delta_j}
\bigl( X(t_k) \otimes Y(t_k) \bigr).
\end{equation}
This defines a time series taking values in the tensor-product Hilbert space $H_A \otimes H_B$.

We centralize this sample by subtracting its empirical mean,
\begin{equation}
\label{L8}
Z(\tau_j) := C(\tau_j) - \bar{C}, 
\qquad
\bar{C} := \frac{1}{N} \sum_{j=1}^N C(\tau_j).
\end{equation}
This notation implies that each $Z(\tau_j)$ is interpreted as a vector in the Hilbert space $H_A \otimes H_B$.

Finally, we define the macro-scale covariance operator by
\begin{equation}
\label{L9}
\hat C =
\frac{1}{N}
\sum_{j=0}^{N-1}
\lvert Z(\tau_j) \rangle \langle Z(\tau_j) \rvert
\;\in\;
\mathcal{L}(H_A \otimes H_B),
\end{equation}
the macro-scale covariance operator associated with the aggregated micro-scale dynamics.

\subsection{Imposing the condition $\mathbb{E}[|C_{\Delta}\rangle]=0$}

{\bf To simplify the notation, in what follows, we assume
\begin{equation}
\label{L5a}
\mathbb{E}[|C_{\Delta}\rangle]=0,
\qquad\text{so that}\qquad
|Z_{\Delta}\rangle=|C_{\Delta}\rangle.
\end{equation}
Consequently, the macro-scale covariance operator is simply
\begin{equation}
\label{L6a}
\widehat{C}
=\mathbb{E}\left[
|C_{\Delta}\rangle\langle C_{\Delta}|
\right]
\in {\cal L}(H_A\otimes H_B).
\end{equation}

A sufficient condition for (\ref{L5a}) is that expectation and time integration commute, that the stochastic processes (X) and (Y) are independent, and that they satisfy the centering condition \eqref{Center}. Indeed,
\begin{equation}
\label{L6b}
\begin{aligned}
\mathbb{E}[|C_{\Delta}\rangle]
&=\frac{1}{|\Delta|}
\int_{\Delta}
\mathbb{E}\left[
|X(t)\rangle\otimes|Y(t)\rangle
\right] \,dt\\
&=\frac{1}{|\Delta|}
\int_{\Delta}
\mathbb{E}[|X(t)\rangle]
\otimes
\mathbb{E}[|Y(t)\rangle] \,dt\\
&=0.
\end{aligned}
\end{equation}
Here, the interchange of expectation and integration is understood under the requisite integrability assumptions.

As we will see, independence of the subquantum stochastic processes does not preclude the generation of entangled states. This illustrates that entanglement of the resulting quantum state is not equivalent to statistical dependence of the underlying stochastic processes. But, generally we need not assume independence. 

This is also the good place to remark that the coincidence of two processes does not imply that the corersponding state \(\rho\) is entangled. The degree of entanglement depends on the stochastic and temporal structure of \(X(t)\), not merely on the 
identity \(X(t)=Y(t)\).
}

\section{Construction of classical stochastic processes behind density operators}
\label{co}

We present a mathematical scheme for generation of entangled states which can be called {\it the subinterval allocation scheme.}
Its advantage is simplicity. In section \ref{DS} we consider general jump processes.   

Now we start with pure states and first consider the most striking example - a maximally entangled state, one of the Bell states. 

\subsection{Bell states from classical correlations}
\label{Bell}

Consider the Bell state
\[
\ket{\Phi^+} = \frac{1}{\sqrt{2}}\left(\ket{0}\otimes\ket{0} + \ket{1}\otimes\ket{1}\right) \in H_A \otimes H_B.
\]

Split a macro window $\Delta$ into two subintervals:
\[
\delta_0 = [0, \Delta/2], \quad \delta_1 = [\Delta/2, \Delta].
\]

Introduce two real valued random variable $\xi_A=\xi_A(\omega)$ and $\xi_B=\xi_B(\omega)$ 
describing macro-randomization in systems $S_A$ and $S_B$ respectively such that 
\begin{equation}
\label{RVS}
\mathbb{E}[\xi_A]=\mathbb{E}[\xi_B]=0, E[\xi_A \xi_B]=0, E[\xi_A^2 \xi_B^2]=2.
\end{equation}
These are uncorrelated random variables, $\rm{Cov}(\xi_A,\xi_B)=0.$ 
In particular, they can be independent random variables with zero mean values and with normalization 
$E[\xi_A^2]=\sqrt{2}, E[\xi_B^2]=\sqrt{2}.$ The random variables can be discrete and take e.g. values $\pm 1.$ 
In this example 
correlations are concentrated at the micro level.  We remark that $\omega$ is a macro-parameter, selection of behaviour of systems 
during time window $\Delta.$

Assign separable Schmidt components to each subinterval:
\[
(X(t,\omega), Y(t,\omega)) =
\begin{cases}
(\xi_A(\omega) \ket{0}, \xi_B(\omega) \ket{0}) & t \in \delta_0,\\
( \xi_A(\omega) \ket{1}, \xi_B(\omega)\ket{1}) & t \in \delta_1.
\end{cases}
\]

Then the micro-level average is
\[
C_{\Delta}(\omega)  = \frac{1}{\Delta} \int_0^\Delta X(t)\otimes Y(t)\, dt
= \xi_A(\omega) \xi_B(\omega)
(\frac{1}{2}\ket{0}\otimes\ket{0} + \frac{1}{2}\ket{1}\otimes\ket{1}) =
\] 
\[
\xi_A(\omega) \xi_B(\omega) \ket{\Phi^+}/\sqrt{2}.
\]
This is a random vector  belonging to $H_A\otimes H_B,$ and its ensemble average (macro-average)  
$E[C_{\Delta}] =E[\xi_A \xi_B] \Phi^+/\sqrt{2}=0,$ so $Z_\Delta = C_{\Delta}.$ Hence,
$$
\hat C = E[|C_{\Delta}\rangle \langle C_{\Delta}|] =
(1/2) \mathbb{E}[\xi_A^2 \xi_B^2] |\Phi^+ \rangle \langle \Phi^+|=|\Phi^+ \rangle \langle \Phi^+|.
$$      
We emphasize once again that subquantum stochastic processes are determined non-uniquely. Above, we presented a simple illustrative example; however, one can construct models with substantially richer internal randomness.

\subsection{Generation of an arbitrary pure state}
\label{pure}

Here we present the simplest scheme of generation of an arbitrary pure state within DCM, similar to the scheme for 
the Bell state $|\Phi^+\rangle,$ more complex stochastic processes can be generated with the scheme 
of section \ref{DS}. 

Each vector $|\psi \rangle \in H_A \otimes H_B$ admits a Schmidt decomposition
\[
|\psi\rangle = \sum_{\ell=1}^{r} s_\ell\;|u_\ell \rangle \otimes |v_\ell\rangle,
\]
where $|u_\ell\rangle \in H_A$ and $|v_\ell\rangle \in H_B,$ $s_\ell > 0.$ 
If $|\psi\rangle$ is a quantum state, then $$\sum_{\ell=1}^{r} s_\ell^2 = 1.$$
(Vectors $(u_\ell)$ are orthonormal as well as vectors $(v_\ell).$)

We now generalize the scheme that was used for generation of the Bell state $|\Phi^+\rangle,$
\begin{enumerate}
\item Partition $\Delta$ into $r$ subintervals $\delta_\ell$ 
of the lengths $(s_\ell/\sum_\ell s_\ell) \Delta.$ 

\item On the $\delta_\ell$-th subinterval, set
\[
X(t,\omega) =  \xi_A(\omega)|u_\ell\rangle,
\qquad
Y(t,\omega) = \xi_B(\omega) |v_\ell\rangle ,
\]
\end{enumerate}
where the random variables $\xi_A(\omega), \xi_B(\omega)$ satisfy conditions similar to conditions (\ref{RVS}),
\begin{equation}
\label{RVS1}
E[\xi_A]=E[\xi_B]=0, E[\xi_A \xi_B]=0, E[\xi_A^2 \xi_B^2]=(\sum_\ell s_\ell)^2.
\end{equation}
In particular, we can consider two independent random variables with zero mean values, such that 
$E[\xi_A^2]=E[\xi_B^2]=(\sum_\ell s_\ell).$  The random variables can be discrete and take e.g. values $\pm 1.$  
Thus, entanglement is generated by microcorrelations.

The  micro-covariance is given by 
\[
C_\Delta( \omega) = \frac{1}{|\Delta|} \int_{\Delta} X(t,\omega)\otimes Y(t,\omega)\,dt
= \frac{\xi_A(\omega) \xi_B(\omega)}{\sum_{\ell=1}^{r} s_\ell}
\sum_{\ell=1}^{r} s_\ell |u_\ell\rangle \otimes |v_\ell \rangle
=\frac{\xi_A( \omega) \xi_B( \omega)}{\sum_{\ell=1}^{r} s_\ell} |\psi\rangle.
\]
we remark that due to our construction, this is a centered random variable, $\mathbb{E}[C_\Delta]=0.$
Now we find its macro-scale covariance
$$
\hat C = \frac{E[ \xi_A^2 \xi_B^2]}{(\sum_\ell s_\ell)^2} |\psi\rangle \langle \psi|= |\psi\rangle \langle \psi|.
$$

\subsection{Generation of mixed states}

Let $H_A$ and $H_B$ be complex Hilbert spaces and let
$\rho$ be a density operator acting  on the tensor product
$H = H_A \otimes H_B$. Consider its spectral decomposition
\[
\rho
=
\sum_{k=1}^{M}
\lambda_k
\ket{\psi_k}\bra{\psi_k},
\qquad
\lambda_k \ge 0,
\qquad
\sum_{k=1}^M \lambda_k = 1,
\qquad
|\psi_k\rangle \in H_A \otimes H_B .
\]
Suppose that there is a random generator selecting time-window $\Delta;$ so an ensemble of time-widows is created. The micro correlations during these time windows generate only the states $\rho_k= |\psi_k\rangle\langle \psi_k|, k=1,..., M.$ 
Assign label $k$ to intervals with the output  $\rho_k : \Delta= \Delta_k.$ Denote the pair of  stochastic processes behind $\rho_k$ as $X_k(t, \omega), Y_k(t, \omega)$ and the corresponding micro correlation as
$S_k(\omega) = C_{\Delta_k}(\omega).$  Suppose that there is a random generator $\eta=\eta(\omega)$ selecting the interval of the $k$-type with probability $\lambda_k.$ As we see from the probabilistic lemma below, if $\eta$ is independent of random variables $S_k,k=1,...,M,$  then this process generates the density operator $\rho.$     

\begin{lemma} [Random selection of stochastic processes]
Let $(\Omega,\mathcal F,\mathbb P)$ be a probability space.
Let $\{S_k\}_{k=1}^M$ be a family of random variables with values in a measurable space $(E,\mathcal E)$,
and let $\eta$ be a discrete random variable taking values in $\{1,\dots,M\}$ with
\[
\mathbb P(\eta = k) = \lambda_k.
\]
Assume that $\eta$ is independent of the family $\{S_k\}_{k=1}^M$.
Define the random variable
\[
S := S_\eta.
\]
Then, for any measurable function $f : E \to L,$ where $L$ is a (finite-dimensional) linear space, such that the expectations exist,
\[
\mathbb E[f(S)]
=
\sum_{k=1}^M
\lambda_k \,
\mathbb E\!\left[f(S_k)\right].
\]
\end{lemma}

In our example, the measurable space $(E,\mathcal E)$ is given by $E=H_A \otimes H_B$ and $\mathcal E$ is the $\sigma$-algebra of Borel subsets of $E;$ random variables  $S_K$ are based on stochastic processes for generation of $\rho_k$ (see, e.g., section \ref{pure} with $|\psi\rangle= |\psi_k\rangle),$ and $f: H_A \otimes H_B \to {\cal L}(H_A \otimes H_B), \; f(z)= |z\rangle\langle z|.$    

\subsection{Generalization of micro-dynamics to jump processes}
\label{DS}

The scheme presented above utilizes locally constant processes defined via characteristic functions on a fixed partition of the macro-window $\Delta$. We now generalize this scheme to general jump processes, where transitions occur at random times.

Let $\{|u_\ell\rangle, |v_\ell\rangle\}_{\ell=1}^r$ be the fixed set of $r$ vector pairs, e.g., Schmidt components for the desired entangled state.

Instead of the piecewise constant processes $X(t, \omega)$ and $Y(t, \omega)$ defined on fixed subintervals, we consider a sequence of random jump times $0 = \tau_0 < \tau_1 < \dots < \tau_{N(\omega)} = \Delta$. Here, $N(\omega)$ is a counting process (such as a Poisson process) representing the number of stochastic events within the window $\Delta$.

Define a mapping $\sigma(j, \omega): \{1, \dots, N(\omega)\} \to \{1, \dots, r\}$ - a random selector. This function randomly assigns one of the $r$ available states to the $j$-th jump interval.

The processes are generalized as:

\begin{equation}
X(t, \omega) = \sum_{j=1}^{N(\omega)} x_j(t, \omega) I_{[\tau_{j-1}, \tau_j)}(t) |u_{\sigma(j, \omega)}\rangle,
\end{equation}
\begin{equation}
Y(t, \omega) = \sum_{j=1}^{N(\omega)} y_j(t, \omega) I_{[\tau_{j-1}, \tau_j)}(t) |v_{\sigma(j, \omega)}\rangle,
\end{equation}
where $I_\delta$ denotes the characteristic function of interval $\delta.$

The micro-covariance  $C_{\Delta}(\omega)$  is now an integral over these random intervals.
Entanglement still emerges from the pathwise alignment of the processes at the micro-scalw, even if the jump times and values are statistically independent at the macro level.

Further coupling with respective domains of theory of classical stochastic processes can be found in preprint \cite{ARX}. 
Our DCM construction can be coupled to the
theory of regime-switching and piecewise-defined processes with examples as
Markov-modulated processes, switching diffusions, piecewise deterministic
Markov processes.

\subsubsection*{Micro-time Consistency}

The concept of micro-time consistency is formalizes an almost-sure constraint on the entire sample path of the jump process.

Let $C \subset H_A \times H_B$ be a set of allowed consistency relations. The pair $(X, Y)$ is said to be micro-time consistent for a general jump process if:
\begin{equation}
P \left( \forall t \in [0, \Delta] : (X(t, \omega), Y(t, \omega)) \in C \right) = 1
\end{equation}
except possibly at the discrete jump instants $\{\tau_j\}$[.

For the Bell states $|\Phi^+\rangle,$ 
 $C = \{ (\lambda_A|j\rangle, \lambda_B|j\rangle) : \lambda_A, \lambda_B \in \mathbb{C}, j \in \{0, 1\} \}.$ This set restricts the pair $(X, Y)$ such that at any given micro-time $t$, both processes must be proportional to the same basis vector ($|0\rangle$ or $|1\rangle$). While this consistency condition is a necessary pathwise constraint, we note that it does not uniquely determine the concrete state (e.g., distinguishing $|\Phi^+\rangle$ from $|\Phi^-\rangle$); that distinction requires the calculation of specific micro-macro correlations.

This approach demonstrates that the DCM framework is not restricted to step functions but applies to any stochastic process where micro-level fluctuations are synchronized according to a global covariance structure. See also \cite{ARX}, appendix B, on further coupling with theory of classical stochastic processes.

\section{Deriving subsystem states from composite systems}

Beginning with a stochastic derivation of the state of a composite system, we now perform transition to the stochastic origin of individual subsystem states. While this approach is somewhat unconventional, it is highly intuitive within the context of quantum mechanics, especially {\it theory of open quantum systems.} In quantum studies, the state $\rho_{AB}$ of a composite system $S_{AB}$ cannot generally be reconstructed from the individual states $\rho_A$ and $\rho_B$ of its components $S_A$ and $S_B.$ However, the states of the subsystems are uniquely determined by the global state through the partial trace operation:
\begin{equation}
\label{consistency7}
\rho_A = \text{Tr}_{H_B} \rho_{AB}, \quad \rho_B = \text{Tr}_{H_A} \rho_{AB}
\end{equation}
A classical probabilistic derivation of these formulas is provided in section \ref{ptrace}, subject to specific constraints on the underlying stochastic processes. For the present discussion, we treat the classical stochastic representation of $\rho_{AB}$ primarily as a conceptual foundation for representing the states of individual systems.

\subsubsection*{The myth of the isolated system: A DCM perspective}

A fundamental question arises within the DCM framework: {\it Do truly isolated quantum physical systems exist, or is an individual system always, by necessity, a subsystem of an encompassing environment?} At first glance, this suggests a potential logical circularity: if an individual system is defined as a marginalized subsystem of a composite, but the composite itself is an individual system at a larger scale, where does the definition ground itself?

\subsubsection*{Breaking the circularity: scale-dependent identity}

DCM avoids this logical circle through its treatment of \textbf{scales of synchronization}. In this framework, the definition of a ``system'' is not an absolute ontological category; rather, it is defined by the \textbf{window of synchronization} ($\Delta$).

\begin{itemize}
    \item {\it The Composite Scale:} A system is viewed as ``Composite'' at the temporal or structural scale where the micro-synchronization between its internal processes ($X(t)$ and $Y(t)$) is explicitly resolved.
    \item {\it The Individual Scale:} That same system becomes an ``Individual'' entity at a higher macro-scale, where internal parts are treated as a unified, emergent fluctuation. 
\end{itemize}
Within the DCM interpretation, the partial trace provides a mathematically precise transition from a composite description to a reduced description. The proposed scale-dependent notion of system identity is intended to address the associated conceptual circularity..

\subsubsection*{Individual systems as marginal residues}

If we follow the logic of the DCM to its conclusion, a truly isolated system is a mathematical idealization. Because the subquantum stochastic processes are likely manifestations of a global field, $X(t)$ is never truly independent. 
In the DCM, the distinction between ``composite'' and ``individual'' is relational rather than ontological. An individual system is essentially the \textbf{marginal residue} of the global field that remains after we lose track of, or purposefully discard, the external correlations (synchronizations) with the environment. Therefore, the ``Quantum State'' $\rho_A$ of an individual system is not a standalone primitive. It is a reduced description of the system's participation in a larger, synchronized whole. The appearance of an isolated system occurs only when the second-order covariance between the system and the rest of the universe becomes negligible or static relative to the macro-observer's window.

This hierarchical view implies that the only truly ``Individual'' system that is not a subsystem would be the Universe itself. Within the DCM, the Universe would be described as a state of perfect global micro-synchronization. The existence of mixed local states and the necessity of the density operator formalism are thus direct consequences of our status as local observers who can only ever perceive a fraction of the total covariance structure.

\section{A stochastic realization of the partial trace identity}
\label{ptrace}

We come to the stochastic representation of the state $\rho_A$ of subsytem $S_A$ of a composite system $S_{AB}$ through stochastic implementation of equality  $\text{Tr}_{H_B} \rho_{AB} = \rho_A$ by explicitly calculating the partial trace and applying micro-synchronization heuristics.

The $H_A\otimes H_B$-valued random variable $C_\Delta$ can be expanded with respect the basis composed of two orthonormal bases, 
 $(|i\rangle) \in H_A$ and  $(|m\rangle) \in H_B:$
$$
C_\Delta= \sum_{im} z_{im} |i\rangle |m\rangle.
$$
The matrix elements of the operator $\hat C_{AB}$ can be expressed as 
$$
(C_{AB})_{im,jk}= \mathbb{E}[z_{im} \bar{z}_{jk}].
$$ 
\subsubsection{Partial trace of covariance}

By definition of the partial trace:
\begin{equation}
(\text{Tr}_{H_B} \hat C_{AB})_{ij} = \sum_m \mathbb{E}[z_{im} \bar{z}_{jm}]
\end{equation}
Substituting the micro-level integral definition of $z_{im}$:
\begin{equation}
(\text{Tr}_{H_B} \hat C_{AB})_{ij} = \frac{1}{\Delta^2} \iint_{\Delta \times \Delta} \mathbb{E} \left[ X_i(t) \bar{X}_j(t') \left( \sum_m Y_m(t) \bar{Y}_m(t') \right) \right] dt \, dt',
\end{equation}
where $X_i(t)= \langle i|X(t) \rangle, Y_m(t)=  \langle m|Y(t) \rangle.$

\subsubsection*{The stochastic reference kernel}
We define the stochastic kernel $K_B(t, t', \omega)$ as the inner product of the micro-signals in $H_B$:
\begin{equation}
K_B(t, t', \omega) = \sum_a Y_a(t, \omega) \bar{Y}_a(t', \omega) = \langle Y(t', \omega) | Y(t, \omega) \rangle_{H_B}.
\end{equation}
The partial trace is then:
\begin{equation}
\label{DI}
(\text{Tr}_{H_B} \hat C_{AB})_{ij} = \frac{1}{\Delta^2} \iint_{\Delta \times \Delta} \mathbb{E} \left[ X_i(t) \bar{X}_j(t') K_B(t, t', \omega) \right] dt \, dt'.
\end{equation}

We will show that, for the special classes of stochastic processes, the double integral in the right-hand side of 
equation (\ref{DI}) is reduced to the integral    
\begin{equation}
\label{sub1}
 (\text{Tr}_{H_B} \hat C_{AB})_{ij} = 
\frac{1}{\Delta} \int_{\Delta} \mathbb{E}[X_i(t) \bar{X}_j(t)] dt. 
\end{equation}
Now consider the subsystem $S_A$ of the composite system $S_{AB}$ as an individual system;  let its micro-macro time randomness is described by a stochastic process 
$X(t).$ Set 
\begin{equation}
\label{sub2}
\hat C_{\Delta;A}(\omega) = \frac{1}{\Delta} \int_{\Delta} |X(t,\omega)\rangle \langle X(t,\omega)| dt,
\end{equation}

This is micro auto-correlation of a stochastic process $X(t).$ It can be treated 
as the operator-valued random variable $\hat C_{\Delta;A}(\omega)$ valued in ${\cal L}(H_A)$  By averaging $\hat C_{\Delta;A}(\omega)$ w.r.t. the macro randomness parameter $\omega$ we obtain the operator 
\begin{equation}
\label{sub3}
\hat C_A= \mathbb{E}[\hat C_{\Delta;A}].
\end{equation}
We define the density  operator for the subsystem $S_A$ of the compositequantum system $S_{AB}$ via its trace-normalization
\begin{equation}
\label{sub4}
\rho_A^{\rm{intrinsic}} = \hat C_A/ \Tr \hat C_A. 
\end{equation}
Our aim is to show that special classes of stochastic processes (see section \ref{consistency} and appendix),
 \begin{equation}
\label{sub47}
\rho_A^{\rm{intrinsic}} = \Tr_{H_B} \rho_{AB} \equiv \rho_A,
\end{equation}
where $ \rho_{AB}=  \hat C_{AB}/ \Tr \hat C_{AB}.$ Thus, for such processes the intrinsic state's definition for $S_A$ coincides with the partial trace definition of the state of $S_A$ as a subsystem of composite system $S_{AB}.$

In fact,  to prove the partial trace equality for density operators, it is sufficient to prove that 
\begin{equation}
\label{sub5}
\hat C_{A} = \Tr_{H_B} \hat C_{AB}.
\end{equation}  
We divide this equality by $\Tr_{H_A\otimes H_B} \hat C_{AB}$ and use that $\Tr_{H_A\otimes H_B} \hat C_{AB}= 
\Tr_{H_A} (\Tr_{H_B}\hat C_{AB}).$

\section{Consistency of subinterval allocation scheme with stochastic realization of partial trace formula}
\label{consistency}

This section examines the consistency between the  ``Subinterval Allocation Scheme'' for construction of entangled states
(section \ref{co}) and stochastic realization of the partial trace formula (section \ref{ptrace}).
 We start  with the Bell state $( |\Phi^+\rangle )$ as a concrete example and then continue to an arbitrary pure state $|\psi\rangle \in H_A\otimes H_B.$

\subsubsection{Bell state $|\Phi^+\rangle $}

In section \ref{Bell}, to construct the Bell state $|\Phi^+\rangle$, we incorporate the random variables $\xi_A(\omega)$ and $\xi_B(\omega)$ which describe macro-randomization in systems $S_A$ and $S_B$. We follow section \ref{Bell} and define the processes as:
$$
(X(t, \omega), Y(t, \omega)) = 
\begin{cases} 
(\xi_A(\omega)|0\rangle, \xi_B(\omega)|0\rangle) & t \in [0, \Delta/2] \\
(\xi_A(\omega)|1\rangle, \xi_B(\omega)|1\rangle) & t \in [\Delta/2, \Delta]
\end{cases}
$$
where $\Delta$ is the macro-window and the (real-valued) random variables
$\xi_A(\omega), \xi_B(\omega)$ satisfy the conditions (\ref{RVS}).

Following the construction in section \ref{ptrace}, we define the stochastic reference kernel $K_B(t, t', \omega)$ as the inner product of the micro-signals in $H_B$:
\begin{equation}
K_B(t, t', \omega) = \langle Y(t', \omega) | Y(t, \omega) \rangle_{H_B} = \xi_B^2(\omega) \left( 1_{\delta_0}(t)1_{\delta_0}(t') + 1_{\delta_1}(t)1_{\delta_1}(t') \right)
\end{equation}
where $\delta_0 = [0, \Delta/2]$ and $\delta_1 = [\Delta/2, \Delta]$ and  $1_{\delta_k}$ is the indicator function of the $k$-th subinterval.

The partial trace $\rho_A$ is calculated via the double integral: 
$$
(\rho_A)_{ij} = \frac{1}{\Delta^2} \iint_{\Delta \times \Delta} \mathbb{E} \left[ X_i(t, \omega) \bar{X}_j(t', \omega) K_B(t, t', \omega) \right] dt \, dt'.
$$
Since $K_B$ vanishes when $t$ and $t'$ are in different subintervals, the integral simplifies. By substituting the specific values for $X(t, \omega)$, the expression for $\rho_A$ becomes:
\begin{equation}
\rho_A = \frac{1}{\Delta^2} \mathbb{E}[\xi_A^2 \xi_B^2] \left( \iint_{\delta_0^2} |0\rangle\langle 0| dt dt' + \iint_{\delta_1^2} |1\rangle\langle 1| dt dt' \right).
\end{equation}
Given that $\mathbb{E}[\xi_A^2 \xi_B^2] = 2$ and the area of each sub-square is $(\Delta/2) \times (\Delta/2) = \Delta^2/4$, we obtain:
\begin{equation}
\rho_A = \frac{1}{\Delta^2} \cdot 2 \cdot \left( \frac{\Delta^2}{4} |0\rangle\langle 0| + \frac{\Delta^2}{4} |1\rangle\langle 1| \right) = \frac{1}{2} (|0\rangle\langle 0| + |1\rangle\langle 1|).
\end{equation}
This yields the standard trace-one mixed state.

In this model, the state of the composite system $\rho_{AB}$ is defined as the double covariance operator $\hat{C}_{AB} = \mathbb{E}[|C_\Delta\rangle \langle C_\Delta|]$. Its trace is given by:
\begin{equation}
\text{Tr} \hat{C}_{AB} = \frac{1}{\Delta} \int_\Delta \mathbb{E}[\|X(t)\|^2 \|Y(t)\|^2] dt = \frac{1}{\Delta} \left( \Delta \cdot \frac{1}{2} \mathbb{E}[\xi_A^2 \xi_B^2] \right) = 1.
\end{equation}
Under the condition of statistical isotropic power, the intrinsic state $\rho_A^{\text{intrinsic}}$ derived from the micro-autocorrelation of $X(t)$ matches the result of the partial trace:
\begin{equation}
\rho_A^{\text{intrinsic}} = \mathbb{E} \left[ \frac{1}{\Delta} \int_\Delta |X(t, \omega)\rangle \langle X(t, \omega)| dt \right] = \frac{1}{2} (|0\rangle\langle 0| + |1\rangle\langle 1|).
\end{equation}
This demonstrates that the normalization of the local state is preserved by the coupling between the macro-randomization ($\xi$) and the micro-scale temporal allocation.

\subsection{Arbitrary pure state}

Consider the Schmidt decomposition of  a pure state $|\psi\rangle \in H_A \otimes H_B,$,
\begin{equation}
\label{PS}
|\psi\rangle=\sum_k s_k |u_k\rangle \otimes |v_k\rangle, \; \sum_k s_k^2=1, s_k >0,
\end{equation}
where $\{|u_k\rangle\}$ and $\{|v_k\rangle\}$ are the Schmidt vectors, $\langle u_k|u_m\rangle=0, \langle v_k|v_m\rangle=0, k \not=m.$  
Set $S= \sum_k s_k >0.$

Let the macro-window $\Delta$ be partitioned into $r$ disjoint subintervals $\Delta = \bigcup_{k=1}^r \delta_k$, where 
$|\delta_k| = (s_k/S) \Delta.$  The micro-processes are defined as:
\begin{equation}
X(t, \omega) = \xi_A(\omega) |u_k\rangle, \quad Y(t, \omega) = \xi_B(\omega) |v_k\rangle \quad \text{for } t \in \delta_k
\end{equation}
 where the (complex valued)  independent random variables
$\xi_A(\omega), \xi_A(\omega)$ satisfy the conditions  (\ref{RVS1}). 

The $Y$-kernel is given by formula:  
\begin{equation}
\label{K}
K_B(t, t', \omega) = \langle Y(t', \omega) | Y(t,\omega) \rangle = \xi_B(\omega)^2 \sum_{k=1}^r 1_{\delta_k}(t) 1_{\delta_k}(t')
\end{equation}
where $1_{\delta_k}$ is the indicator function of the $k$-th subinterval.

The partial trace of the macro-covariance is given by the double integral:
\begin{equation}
\label{K1}
(\rho_A)_{ij} = \frac{1}{\Delta^2} \iint_{\Delta \times \Delta} E[X_i(t) \bar{X}_j(t') K_B(t, t',\omega)] \, dt \, dt'
\end{equation}
Substituting the kernel (\ref{K}) we obtain: 
\begin{equation}
(\rho_A)_{ij} = \frac{1}{\Delta^2 } \sum_{k=1}^r \iint_{\delta_k \times \delta_k} E[\xi_B(\omega)^2 X_i(t) \bar{X}_j(t')]  \, dt \, dt'
\end{equation}
Since $X(t)$ is constant vector  ($|u_k\rangle$) on each subinterval $\delta_k,$ we have
\begin{equation}
\label{K2}
(\rho_A)_{ij} = E [\xi_A(\omega)^2 \xi_B(\omega)^2]\frac{1}{\Delta^2} \sum_{k=1}^r  (u_k)_i (\bar{u}_k)_j \cdot |\delta_k|^2= 
\end{equation}
$$
= \frac{S^2}{\Delta^2} \sum_{k=1}^r  (u_k)_i (\bar{u}_k)_j \cdot ( \Delta s_k/S)^2 
$$
Thus
\begin{equation}
\rho_A = \sum_{k=1}^r s_k^2 |u_k\rangle\langle u_k|
\end{equation}
By the normalization of the joint state $Z = \sum s_k (u_k \otimes v_k)$, this result is identical to the partial trace derived from the standard quantum formalism.

\section{Concluding Discussion: The Relational Nature of Systems}

DCM provides a stochastic construction of quantum density operators from classical processes. In this framework, the quantum state is obtained as the normalized covariance operator of a random vector in the tensor-product space of a composite system. The model connects the temporal structure of underlying stochastic processes with the statistical description of quantum states, providing a framework for reconstructing both pure and mixed states. The results of this paper establish the mathematical foundations of this construction, demonstrate how the subinterval allocation scheme generates prescribed states, and examine its consistency with the partial-trace operation. Exact identities are obtained for particular allocations, while Appendix A establishes sufficient conditions for approximate partial-trace identities in a limiting regime.

\paragraph{Scale, Synchronization, and Entanglement.}

A central feature of the DCM is the distinction between the micro-time structure of stochastic processes and the macro-time statistical description obtained by averaging over a temporal window. The resulting covariance operator, after normalization, has the mathematical properties of a quantum density operator.

The subinterval allocation scheme illustrates how this construction can reproduce entangled states. By assigning components of a target state to different subintervals and imposing suitable moment conditions on the random amplitudes, the model generates the desired density operator. This provides a stochastic realization of quantum-state representations through the temporal organization of classical processes.

The partial-trace results further connect the DCM with the standard quantum formalism. For the constructions considered here, the intrinsic covariance of a subsystem agrees with the partial trace of the composite covariance. More general approximate identities hold under the conditions established in Appendix A. These results should not be understood as universal identities for arbitrary stochastic processes, but as consequences of specific constructions and assumptions.

\paragraph{Relational Identity and the Description of Subsystems}

The DCM motivates a relational perspective on the distinction between composite systems and their subsystems. A subsystem state is connected to the statistical structure of the larger system through the partial trace or, under the conditions established in this paper, through an intrinsic covariance construction. This suggests that the description of a system may depend on the level at which its correlations with other degrees of freedom are represented. A collection of components may be treated as composite at one level and described through a reduced state at another. The DCM provides a stochastic framework for expressing this change of description.

Stronger ontological claims—for example, that individual systems are necessarily marginal residues of a global field, or that the Universe is the only truly individual system—remain interpretations rather than mathematical conclusions established here. The present results support a relational description of states without requiring a unique metaphysical interpretation.

\paragraph{Broader Implications and Future Directions}

The construction of density operators, including entangled states, from classical stochastic processes may be useful not only in quantum foundational studies, but also in quantum-like (inspired) modeling and simulation. Potential applications include data analysis, signal processing, and models of cognition and decision making, where density operators may represent structured correlations generated by underlying stochastic processes. Whether this approach provides useful models in specific domains remains an open question.

In summary, the DCM establishes a stochastic route from classical processes to quantum-state representations. Its principal contribution in this paper is the explicit construction of density operators and the demonstration, under specified conditions, of their consistency with the partial-trace formalism. It thus provides a framework for investigating how quantum-like statistical structures can arise from classical stochastic descriptions, while leaving broader foundational and application-related questions open.

\newenvironment{myappendices}
  {\appendix}
  {}
\begin{myappendices}

\section{Approximate partial trace formula}

Here we describe a class of stochastic processes approximately reproducing the partial trace formula.
This formula is recovered in the $\epsilon \to 0$ limit where $\epsilon>0$ is a small parameter describing 
the diagonal concentration of the double covariances. We proceed in two steps. First, we consider sufficient 
one-sided conditions (for the $Y$-process), stated in Theorem \ref{aT1}, which guarantee an approximate partial-trace formula for subsystem $A$ (Section \ref{A1}). In Section \ref{A2}, Gaussian covariations show that this class of $Y$-processes is non-empty and in fact rich. We then formulate two-sided conditions in Theorem \ref{aT2}, guaranteeing an approximate partial-trace formula for both subsystems $A$ and $B$ of the composite system $AB$.

\subsection{Sufficient conditions for the $Y$-process}
\label{A1}

Before stating the theorem, introduce the following notation. Let
$H_A=\mathbb{C}^{d_A}$ and $H_B=\mathbb{C}^{d_B}$, and let
$X:[0,\Delta]\times\Omega\to H_A$ and
$Y_\varepsilon:[0,\Delta]\times\Omega\to H_B$ be stochastic processes.
Write
\[
C_X^{ij}(t,t'):=\mathbb{E}[X_i(t)\overline{X_j(t')}],
\qquad
(C_A)_{ij}:=\frac1\Delta\int_0^\Delta C_X^{ij}(t,t)\,dt,
\]
and assume that each $C_X^{ij}$ is continuous on $[0,\Delta]^2$ and
$\Tr C_A>0$. Define the expected correlation kernel of $Y_\varepsilon$ by
\[
R_\varepsilon(t,t')
:=\mathbb{E}[\langle Y_\varepsilon(t'),Y_\varepsilon(t)\rangle_{H_B}]
=\sum_{m=1}^{d_B}
\mathbb{E}[Y_{\varepsilon,m}(t)\overline{Y_{\varepsilon,m}(t')}].
\]
Let $C_{AB}(\varepsilon)$ be the composite covariance matrix with entries
\[
(C_{AB}(\varepsilon))_{im,jk}
:=\frac1{\Delta^2}\int_0^\Delta\!\!\int_0^\Delta
\mathbb{E}\!\left[
X_i(t)\overline{X_j(t')}
Y_{\varepsilon,m}(t)\overline{Y_{\varepsilon,k}(t')}
\right]\,dt'\,dt,
\]
and set
\[
q_\varepsilon:=\frac1{\Delta^2}
\int_0^\Delta\!\!\int_0^\Delta
R_\varepsilon(t,t')\,dt'\,dt.
\]

\begin{theorem}[One-side partial trace formula]
\label{aT1}
Suppose that the following conditions hold:
\begin{enumerate}
\item[(i)] \textbf{Factorization.} For all indices and times,
\[
\begin{aligned}
&\mathbb{E}\!\left[
X_i(t)\overline{X_j(t')}
Y_{\varepsilon,m}(t)\overline{Y_{\varepsilon,k}(t')}
\right]\\
&\qquad=C_X^{ij}(t,t')\,
\mathbb{E}\!\left[Y_{\varepsilon,m}(t)
\overline{Y_{\varepsilon,k}(t')}\right].
\end{aligned}
\]

\item[(ii)] \textbf{Kernel concentration.} We have $q_\varepsilon>0$ and,
for every continuous $f:[0,\Delta]^2\to\mathbb C$,
\[
\lim_{\varepsilon\to0}
\frac1{q_\varepsilon}
\int_0^\Delta\!\!\int_0^\Delta
f(t,t')R_\varepsilon(t,t')\,dt'\,dt
=\Delta\int_0^\Delta f(t,t)\,dt.
\]
\end{enumerate}
Then
\[
\boxed{
\lim_{\varepsilon\to0}
\Tr_{H_B}\!\left(
\frac{C_{AB}(\varepsilon)}{\Tr C_{AB}(\varepsilon)}
\right)
=\frac{C_A}{\Tr C_A}.
}
\]
\end{theorem}

\begin{proof}
Put $N_\varepsilon:=\Tr_{H_B}C_{AB}(\varepsilon)$. By (i),
\[
\begin{aligned}
(N_\varepsilon)_{ij}
&=\sum_{m=1}^{d_B}(C_{AB}(\varepsilon))_{im,jm}\\
&=\frac1{\Delta^2}
\int_0^\Delta\!\!\int_0^\Delta
C_X^{ij}(t,t')R_\varepsilon(t,t')\,dt'\,dt.
\end{aligned}
\]
Using (ii) with $f=C_X^{ij}$ gives
\[
\begin{aligned}
\lim_{\varepsilon\to0}\frac{(N_\varepsilon)_{ij}}{q_\varepsilon}
&=\frac1{\Delta^2}\Delta
\int_0^\Delta C_X^{ij}(t,t)\,dt\\
&=(C_A)_{ij}.
\end{aligned}
\]
Hence $N_\varepsilon/q_\varepsilon\to C_A$. It follows that
\[
\frac{\Tr N_\varepsilon}{q_\varepsilon}\to\Tr C_A>0.
\]
Since $\Tr C_{AB}(\varepsilon)=\Tr N_\varepsilon$,
\[
\begin{aligned}
\Tr_{H_B}\!\left(
\frac{C_{AB}(\varepsilon)}{\Tr C_{AB}(\varepsilon)}
\right)
&=\frac{N_\varepsilon}{\Tr N_\varepsilon}\\
&=\frac{N_\varepsilon/q_\varepsilon}
{\Tr N_\varepsilon/q_\varepsilon}
\longrightarrow\frac{C_A}{\Tr C_A}.
\end{aligned}
\]
\end{proof}

\begin{remark}[Amplitude rescaling]
If $Y_\varepsilon=a_\varepsilon Z_\varepsilon$ with $a_\varepsilon\ne0$
independent of time, then both $R_\varepsilon$ and $q_\varepsilon$ are
multiplied by $|a_\varepsilon|^2$. Thus condition (ii) and the normalized
composite covariance are invariant under this rescaling. The theorem
therefore determines the normalized reduced limit, not the absolute
amplitude of the auxiliary process.
\end{remark}

\begin{remark}[Scope]
Condition (ii) is a substantive assumption about the correlation kernels
generated by $Y_\varepsilon$ and must be verified for the chosen model.
The theorem does not by itself derive this concentration property.
\end{remark}

\subsection{The delayed Gaussian auxiliary process}
\label{A2}

Fix $\Delta>0$ and $\tau>0$. Let
\[
r_\varepsilon(s)
:=\varepsilon\delta_\varepsilon(s)
=\sqrt{\frac{\varepsilon}{2\pi}}
\exp\left(-\frac{s^2}{2\varepsilon}\right).
\]

For each $\varepsilon>0$, let $\xi_\varepsilon$ be a centered proper
complex Gaussian process on $\mathbb R$ with covariance
\[
\mathbb E\!\left[\xi_\varepsilon(t)
 \overline{\xi_\varepsilon(t')}\right]
=r_\varepsilon(t-t').
\]
Such a process exists because $r_\varepsilon$ is a continuous
positive-definite Gaussian covariance kernel.

Choose nonnegative smooth functions $a,b:[0,\Delta]\to[0,1]$ such that
\[
a(t)^2+b(t)^2=1,\qquad 0\leq t\leq\Delta.
\]
Define the $H_B=\mathbb C^2$-valued Gaussian process
\begin{equation}\label{eq:Y-def}
Y_\varepsilon(t)=
\begin{pmatrix}
b(t)\xi_\varepsilon(t)\\
a(t)\xi_\varepsilon(t-\tau)
\end{pmatrix}.
\end{equation}
The two components are correlated because they are delayed versions of
the same Gaussian process. Their matrix-valued covariance kernel is
\begin{equation}\label{eq:KY}
K_Y^\varepsilon(t,t')=
\begin{pmatrix}
b(t)b(t')r_\varepsilon(t-t')
&b(t)a(t')r_\varepsilon(t-t'+\tau)\\[2pt]
a(t)b(t')r_\varepsilon(t-t'-\tau)
&a(t)a(t')r_\varepsilon(t-t')
\end{pmatrix}.
\end{equation}
In particular, the off-diagonal entries have correlation peaks at
time differences $-\tau$ and $+\tau$.

\subsection{Verification of the one-sided theorem's assumptions}
\label{A3}

The expected scalar correlation kernel is the trace of
$K_Y^\varepsilon$:
\[
R_\varepsilon(t,t')=\operatorname{Tr}K_Y^\varepsilon(t,t')
=r_\varepsilon(t-t')c(t,t'),
\qquad
c(t,t'):=b(t)b(t')+a(t)a(t').
\]
The function $c$ is continuous and satisfies
\[
c(t,t)=a(t)^2+b(t)^2=1.
\]
Since $a,b$ are nonnegative, $c(t,t')\geq0$. Define
\[
q_\varepsilon=\frac{1}{\Delta^2}
\int_0^\Delta\!\!\int_0^\Delta
R_\varepsilon(t,t')\,dt'\,dt.
\]
Equivalently,
\[
q_\varepsilon=\frac{\varepsilon}{\Delta^2}
\int_0^\Delta\!\!\int_0^\Delta
c(t,t')\delta_\varepsilon(t-t')\,dt'\,dt>0.
\]
Strict positivity follows because $c(t,t)=1$ and $c$ is continuous,
so the integrand is positive in a neighbourhood of the diagonal.

The Gaussian approximate-identity property gives
\[
\int_0^\Delta\!\!\int_0^\Delta
c(t,t')\delta_\varepsilon(t-t')\,dt'\,dt
\longrightarrow
\int_0^\Delta c(t,t)\,dt=\Delta.
\]
Hence
\begin{equation}\label{eq:q-asymptotic}
q_\varepsilon\sim\frac{\varepsilon}{\Delta}.
\end{equation}
For every continuous $h:[0,\Delta]^2\to\mathbb C$,
\begin{align*}
\frac{1}{q_\varepsilon}
\int_0^\Delta\!\!\int_0^\Delta
h(t,t')R_\varepsilon(t,t')\,dt'\,dt
&=\Delta^2\frac{\displaystyle
\int_0^\Delta\!\!\int_0^\Delta
h(t,t')c(t,t')\delta_\varepsilon(t-t')\,dt'\,dt}
{\displaystyle
\int_0^\Delta\!\!\int_0^\Delta
c(t,t')\delta_\varepsilon(t-t')\,dt'\,dt}\\
&\longrightarrow
\Delta\int_0^\Delta h(t,t)\,dt.
\end{align*}
Thus $Y_\varepsilon$ satisfies exactly the kernel-concentration
assumption of the one-sided reduced-covariance theorem.

Now let $X:[0,\Delta]\times\Omega\to H_A=\mathbb C^{d_A}$ be any
fixed admissible process, independent of $\xi_\varepsilon$, whose
component second-moment kernels
\[
K_X^{ij}(t,t')=\mathbb E[X_i(t)\overline{X_j(t')}]
\]
are continuous and whose intrinsic covariance
\[
C_A=\frac{1}{\Delta}\int_0^\Delta K_X(t,t)\,dt
\]
has positive trace. Independence gives the factorization assumption
of the theorem. Applying the theorem with $h=K_X^{ij}$ proves
\begin{equation}\label{eq:universal-one-sided}
\boxed{
\operatorname{Tr}_{H_B}
\frac{C_{AB}^\varepsilon}{\operatorname{Tr}C_{AB}^\varepsilon}
\longrightarrow
\frac{C_A}{\operatorname{Tr}C_A}.
}
\end{equation}
Consequently, this same delayed Gaussian family gives the one-sided
partial-trace limit for every fixed $X$ in the theorem's admissible
class, provided $X$ is independent of $Y_\varepsilon$. The assertion
is for fixed $X$; it does not cover arbitrary $\varepsilon$-dependent
families $X^\varepsilon$ without a uniformity assumption.

\subsection{A choice of $X$ yielding an entangled composite state}
\label{A4}

We now choose $X$ so that the composite covariance is
entangled. Let $|0\rangle,|1\rangle$ be the standard basis
of $H_A=H_B=\mathbb{C}^2$. Assume that $0<\tau<\Delta$.
Choose an interval $J$ whose closure is contained in
$(0,\Delta-\tau)$ and whose diameter is smaller than $\tau$.
Let $f\in C_c^\infty(J)$ be real-valued and nonzero, and set
\[
L:=\int_J f(t)^2\,dt>0.
\]
The supports of $f(t)$ and $f(t-\tau)$ are disjoint, since
$\operatorname{diam}(J)<\tau$.

Choose the smooth masks
$a,b$ above so that
\[
b(t)=1,\ a(t)=0\quad\text{on }\operatorname{supp}f,
\qquad
b(t)=0,\ a(t)=1\quad\text{on }\operatorname{supp}f+\tau.
\]
This is possible by choosing a smooth angle $\theta$ with
$a=\sin\theta$, $b=\cos\theta$, and assigning $\theta=0$ and
$\theta=\pi/2$ on the two disjoint supports.

Let $\zeta$ be a centered complex random variable independent of
$\xi_\varepsilon$, with $\mathbb E|\zeta|^2=1$. Define
\begin{equation}\label{eq:X-def}
X(t)=\zeta\bigl(f(t)\ket{0}+f(t-\tau)\ket{1}\bigr).
\end{equation}
Then $X$ is independent of $Y_\varepsilon$, its covariance kernel is
continuous, and its intrinsic covariance is
\[
C_A=\frac{1}{\Delta}\int_0^\Delta
\mathbb E[X(t)X(t)^*]\,dt
=\frac{L}{\Delta}I_2.
\]
In particular, $C_A/\operatorname{Tr}C_A=I_2/2$.

To compute the composite covariance, introduce the random vector
\[
Z_\Delta^\varepsilon:=\frac{1}{\Delta}
\int_0^\Delta X(t)\otimes Y_\varepsilon(t)\,dt.
\]
Under the covariance convention used above,
$C_{AB}^\varepsilon=\mathbb E[
\ket{Z_\Delta^\varepsilon}\bra{Z_\Delta^\varepsilon}]$.
On $\operatorname{supp}f$, the integrand is
\[
X(t)\otimes Y_\varepsilon(t)
=\zeta f(t)\xi_\varepsilon(t)\ket{00},
\]
whereas on $\operatorname{supp}f+\tau$ it is
\[
X(t)\otimes Y_\varepsilon(t)
=\zeta f(t-\tau)\xi_\varepsilon(t-\tau)\ket{11}.
\]
It vanishes outside these two supports. Changing variables
$u=t-\tau$ in the second integral gives the exact identity
\[
Z_\Delta^\varepsilon
=\frac{\zeta}{\Delta}
\left(\int_J f(u)\xi_\varepsilon(u)\,du\right)
\bigl(\ket{00}+\ket{11}\bigr).
\]
Set
\[
V_\varepsilon:=\mathbb E\left|
\int_J f(u)\xi_\varepsilon(u)\,du
\right|^2
=\int_J\!\!\int_J
f(u)f(v)r_\varepsilon(u-v)\,du\,dv.
\]
Because $r_\varepsilon$ is a strictly positive-definite Gaussian
kernel and $f\not\equiv0$, $V_\varepsilon>0$. Also,
\[
\frac{V_\varepsilon}{\varepsilon}
=\int_J\!\!\int_J
f(u)f(v)\delta_\varepsilon(u-v)\,du\,dv
\longrightarrow L.
\]
Writing
\[
S:=\bigl(\ket{00}+\ket{11}\bigr)
  \bigl(\bra{00}+\bra{11}\bigr),
\]
we obtain the exact composite covariance
\[
C_{AB}^\varepsilon=\frac{V_\varepsilon}{\Delta^2}S.
\]
Consequently,
\[
\boxed{
\rho_{AB}^\varepsilon:=
\frac{C_{AB}^\varepsilon}{\operatorname{Tr}C_{AB}^\varepsilon}
=\frac12\bigl(\ket{00}+\ket{11}\bigr)
\bigl(\bra{00}+\bra{11}\bigr)
=\ket{\Phi^+}\bra{\Phi^+}
}
\]
for every $\varepsilon>0$, where
\[
\ket{\Phi^+}=\frac{\ket{00}+\ket{11}}{\sqrt2}.
\]
Thus the normalized composite state is exactly a maximally entangled
Bell state, not merely asymptotic to one. Its partial trace over
$H_B$ is $I_2/2=C_A/\operatorname{Tr}C_A$, in agreement with
\eqref{eq:universal-one-sided}.

\begin{remark}
The universal conclusion \eqref{eq:universal-one-sided} uses only the
trace of the covariance kernel of $Y_\varepsilon$. Its delayed
off-diagonal correlations do not alter that scalar concentration
condition, but they allow the particular choice \eqref{eq:X-def} to
produce coherence between $\ket{00}$ and $\ket{11}$. This example
proves the one-sided reduction and entanglement simultaneously; it
does not, by itself, assert that the reverse partial-trace theorem
holds for every admissible $Y_\varepsilon$.
\end{remark}

\newpage

\paragraph{Definitions.}
Let \(H_A=\mathbb C^{d_A}\) and \(H_B=\mathbb C^{d_B}\).

For each $\epsilon>0$, let $X^\epsilon(t)\in H_A$ and
$Y^\epsilon(t)\in H_B$, $t\in[0,\Delta]$, be independent,
centered stochastic processes 
\[
\mathbb{E}[X^\epsilon(t)]=\mathbb{E}[Y^\epsilon(t)]=0
\quad\text{for every }t\in[0,\Delta]\text{ and }\epsilon>0,
\]
with continuous covariance kernels,
\[
\begin{aligned}
K_X^\varepsilon(t,t')
&=\mathbb E\!\left[
X^\varepsilon(t)X^\varepsilon(t')^*
\right],\\
K_Y^\varepsilon(t,t')
&=\mathbb E\!\left[
Y^\varepsilon(t)Y^\varepsilon(t')^*
\right].
\end{aligned}
\]
Here \(^*\) denotes conjugate transpose. Define the intrinsic covariances
\[
\begin{aligned}
C_A^\varepsilon
&=\frac{1}{\Delta}\int_0^\Delta
K_X^\varepsilon(t,t)\,dt,\\
C_B^\varepsilon
&=\frac{1}{\Delta}\int_0^\Delta
K_Y^\varepsilon(t,t)\,dt.
\end{aligned}
\]
Write
\[
R_X^\varepsilon(t,t')=\operatorname{Tr}K_X^\varepsilon(t,t'),
\qquad
R_Y^\varepsilon(t,t')=\operatorname{Tr}K_Y^\varepsilon(t,t').
\]
The composite covariance \(C_{AB}^\varepsilon\) is defined by
\begin{equation}
\label{eq:composite-covariance}
\begin{aligned}
(C_{AB}^\varepsilon)_{im,jk}
=\frac{1}{\Delta^2}
\int_0^\Delta\!\!\int_0^\Delta
&(K_X^\varepsilon)_{ij}(t,t')\\[-2pt]
&\quad{}\times
(K_Y^\varepsilon)_{mk}(t,t')\,dt'\,dt .
\end{aligned}
\end{equation}
Let
\[
N_A^\varepsilon=\operatorname{Tr}_{H_B}C_{AB}^\varepsilon,
\qquad
N_B^\varepsilon=\operatorname{Tr}_{H_A}C_{AB}^\varepsilon.
\]

\begin{theorem}[Simultaneous partial-trace limits]
\label{aT2}
Assume that in matrix norm: 
\[
C_A^\varepsilon\longrightarrow C_A,
\qquad
C_B^\varepsilon\longrightarrow C_B
\quad\text{as }\varepsilon\to0,
\]
where \(\operatorname{Tr}C_A>0\) and \(\operatorname{Tr}C_B>0\).\footnote{Since \(C_A^\varepsilon\) and \(C_B^\varepsilon\) are Hermitian positive semidefinite, their limits \(C_A\) and \(C_B\) are also Hermitian positive semidefinite.}.
Assume that there exist positive scalars
\(\alpha_A^\varepsilon\) and \(\alpha_B^\varepsilon\) such that
\begin{equation}
\label{eq:compatibility-A}
\lim_{\varepsilon\to0}
\frac{1}{\alpha_A^\varepsilon}
\left\|
\frac{1}{\Delta^2}
\int_0^\Delta\!\!\int_0^\Delta
K_X^\varepsilon(t,t')R_Y^\varepsilon(t,t')
\,dt'\,dt
-\alpha_A^\varepsilon C_A^\varepsilon
\right\|_1=0
\end{equation}
and
\begin{equation}
\label{eq:compatibility-B}
\lim_{\varepsilon\to0}
\frac{1}{\alpha_B^\varepsilon}
\left\|
\frac{1}{\Delta^2}
\int_0^\Delta\!\!\int_0^\Delta
R_X^\varepsilon(t,t')K_Y^\varepsilon(t,t')
\,dt'\,dt
-\alpha_B^\varepsilon C_B^\varepsilon
\right\|_1=0.
\end{equation}
Then both normalized partial traces converge:
\begin{equation}
\label{eq:limit-A}
\operatorname{Tr}_{H_B}
\frac{C_{AB}^\varepsilon}
{\operatorname{Tr}C_{AB}^\varepsilon}
\longrightarrow
\frac{C_A}{\operatorname{Tr}C_A},
\end{equation}
and
\begin{equation}
\label{eq:limit-B}
\operatorname{Tr}_{H_A}
\frac{C_{AB}^\varepsilon}
{\operatorname{Tr}C_{AB}^\varepsilon}
\longrightarrow
\frac{C_B}{\operatorname{Tr}C_B}.
\end{equation}
\end{theorem}

\begin{proof}
Taking the partial traces in
\eqref{eq:composite-covariance} gives the exact identities
\begin{align}
N_A^\varepsilon
&=\frac{1}{\Delta^2}
\int_0^\Delta\!\!\int_0^\Delta
K_X^\varepsilon(t,t')R_Y^\varepsilon(t,t')
\,dt'\,dt,
\label{eq:partial-A}\\
N_B^\varepsilon
&=\frac{1}{\Delta^2}
\int_0^\Delta\!\!\int_0^\Delta
R_X^\varepsilon(t,t')K_Y^\varepsilon(t,t')
\,dt'\,dt.
\label{eq:partial-B}
\end{align}
Therefore, the assumptions
\eqref{eq:compatibility-A}--\eqref{eq:compatibility-B} imply
\[
\frac{N_A^\varepsilon}{\alpha_A^\varepsilon}
-C_A^\varepsilon\longrightarrow0,
\qquad
\frac{N_B^\varepsilon}{\alpha_B^\varepsilon}
-C_B^\varepsilon\longrightarrow0
\]
in trace norm. Since \(C_A^\varepsilon\to C_A\) and
\(C_B^\varepsilon\to C_B\), it follows that
\[
\frac{N_A^\varepsilon}{\alpha_A^\varepsilon}
\longrightarrow C_A,
\qquad
\frac{N_B^\varepsilon}{\alpha_B^\varepsilon}
\longrightarrow C_B.
\]
Taking traces and using the strictly positive limiting traces yields
\[
\frac{\operatorname{Tr}N_A^\varepsilon}
{\alpha_A^\varepsilon}
\longrightarrow\operatorname{Tr}C_A>0,
\qquad
\frac{\operatorname{Tr}N_B^\varepsilon}
{\alpha_B^\varepsilon}
\longrightarrow\operatorname{Tr}C_B>0.
\]
Normalizing these convergences gives
\[
\frac{N_A^\varepsilon}{\operatorname{Tr}N_A^\varepsilon}
\longrightarrow\frac{C_A}{\operatorname{Tr}C_A},
\qquad
\frac{N_B^\varepsilon}{\operatorname{Tr}N_B^\varepsilon}
\longrightarrow\frac{C_B}{\operatorname{Tr}C_B}.
\]
Finally, the exact trace identities for partial traces give
\[
\operatorname{Tr}N_A^\varepsilon
=\operatorname{Tr}C_{AB}^\varepsilon
=\operatorname{Tr}N_B^\varepsilon.
\]
Consequently, the two normalized matrices on the left are exactly the
normalized partial traces in \eqref{eq:limit-A} and
\eqref{eq:limit-B}. The positive limiting traces also ensure that the
normalization is well-defined for all sufficiently small
\(\varepsilon\). This proves both limits.
\end{proof}

\begin{corollary}[Trace-scale compatibility]
Under the assumptions of Theorem~\ref{aT2},
\begin{equation}
\label{eq:trace-scale-compatibility}
\frac{\alpha_A^\varepsilon\operatorname{Tr}C_A}
{\alpha_B^\varepsilon\operatorname{Tr}C_B}
\longrightarrow 1.
\end{equation}
\end{corollary}

\begin{proof}
The exact identity
\(\operatorname{Tr}N_A^\varepsilon
=\operatorname{Tr}N_B^\varepsilon\), together with the two trace
limits in the theorem's proof, gives
\[
\frac{\alpha_A^\varepsilon\operatorname{Tr}C_A}
{\alpha_B^\varepsilon\operatorname{Tr}C_B}
=
\frac{\alpha_A^\varepsilon\operatorname{Tr}C_A}
{\operatorname{Tr}N_A^\varepsilon}
\frac{\operatorname{Tr}N_B^\varepsilon}
{\alpha_B^\varepsilon\operatorname{Tr}C_B}
\longrightarrow 1.
\]
Thus \eqref{eq:trace-scale-compatibility} is a consequence of the
two compatibility assumptions, not an independent hypothesis.
\end{proof}

\paragraph{Scope.}
The compatibility conditions
\eqref{eq:compatibility-A}--\eqref{eq:compatibility-B} are sufficient
conditions for the simultaneous limits. They are not derived here from
ordinary scalar approximate-identity concentration. Establishing them
for a concrete class of matrix-valued covariance kernels, particularly
for a nonfactorizable composite limit, requires a separate argument.

\subsection{Verification of the simultaneous compatibility conditions}
\label{A5}

We verify the two compatibility conditions of Theorem~\ref{aT2} for the entangled
example constructed in Section \ref{A4}. We retain the process $X$, the quantities
$L$ and $V_\epsilon$, and the exact expression for
$C_{AB}^{\epsilon}$ obtained there. For $Y_\epsilon$, we use the delayed
Gaussian construction of Section~9, with $\tau=\Delta/2$.

We now specify the smooth masks so that they satisfy both the support
constraints in Section~11 and the balancing condition required below.
Write
\[
a(t)=\sin\theta(t),\qquad b(t)=\cos\theta(t).
\]
Since $\operatorname{supp}f$ is compactly contained in
$(0,\Delta/2)$, choose a smooth function
$\theta:[0,\Delta/2]\to[0,\pi/2]$ such that
$\theta=0$ on $\operatorname{supp}f$ and $\theta=\pi/4$
in neighborhoods of both endpoints $0$ and $\Delta/2$.

Extend $\theta$ to $[0,\Delta]$ by imposing
\[
\theta(t+\Delta/2)=\frac{\pi}{2}-\theta(t),
\qquad 0\leq t\leq\Delta/2.
\]
Because $\theta=\pi/4$ near the endpoints of the first half-interval,
this extension is smooth at $t=\Delta/2$. Moreover, $\theta=0$ on
$\operatorname{supp}f$ implies $\theta=\pi/2$ on
$\operatorname{supp}f+\Delta/2$. Consequently, the support constraints
of Section~11 are satisfied.

The symmetry of the extension gives
\[
a(t+\Delta/2)^2=b(t)^2,\qquad
b(t+\Delta/2)^2=a(t)^2,
\quad 0\leq t\leq\Delta/2.
\]
Therefore,
\begin{align*}
\int_0^\Delta a(t)^2\,dt
&=\int_0^{\Delta/2}\bigl(a(t)^2+b(t)^2\bigr)\,dt
=\frac{\Delta}{2},\\
\int_0^\Delta b(t)^2\,dt
&=\int_0^{\Delta/2}\bigl(b(t)^2+a(t)^2\bigr)\,dt
=\frac{\Delta}{2},
\end{align*}
where we used $a(t)^2+b(t)^2=1$. Hence both the support constraints
and the balanced-mask condition hold.

\subsection*{The limiting intrinsic covariance of \(Y_\varepsilon\)}

Write \(c_0=1/\sqrt{2\pi}\), so
\(r_\varepsilon(s)=c_0e^{-s^2/(2\varepsilon^2)}\). From the covariance
kernel in Section~4.1,
\[
K_Y^\varepsilon(t,t)=
\begin{pmatrix}
c_0b(t)^2&r_\varepsilon(\tau)a(t)b(t)\\
r_\varepsilon(\tau)a(t)b(t)&c_0a(t)^2
\end{pmatrix}.
\]
The balanced-mask identities therefore give
\[
C_B^\varepsilon
=\frac{1}{\Delta}\int_0^\Delta K_Y^\varepsilon(t,t)\,dt
=\frac{c_0}{2}I_2+
d_\varepsilon
\begin{pmatrix}0&1\\1&0\end{pmatrix},
\]
where
\[
d_\varepsilon
:=\frac{r_\varepsilon(\tau)}{\Delta}
  \int_0^\Delta a(t)b(t)\,dt.
\]
Since \(\tau>0\),
\[
r_\varepsilon(\tau)
=c_0e^{-\tau^2/(2\varepsilon^2)}
\longrightarrow0,
\qquad d_\varepsilon\longrightarrow0.
\]
Consequently,
\begin{equation}
\label{eq:CB-limit-section5}
C_B^\varepsilon\longrightarrow C_B:=\frac{c_0}{2}I_2,
\qquad \operatorname{Tr}C_B=c_0>0.
\end{equation}
By Section~3,
\[
C_A^\varepsilon=C_A=\frac{L}{\Delta}I_2,
\qquad \operatorname{Tr}C_A=\frac{2L}{\Delta}>0.
\]
Thus both limiting intrinsic states are \(I_2/2\).

\subsection*{Verification of the two compatibility conditions}

Let
\[
N_A^\varepsilon=\operatorname{Tr}_{H_B}C_{AB}^\varepsilon,
\qquad
N_B^\varepsilon=\operatorname{Tr}_{H_A}C_{AB}^\varepsilon.
\]
The exact composite-covariance formula proved in Section~3 implies
\[
N_A^\varepsilon=N_B^\varepsilon
=\frac{V_\varepsilon}{\Delta^2}I_2,
\qquad V_\varepsilon>0.
\]

For the first compatibility condition, choose
\[
\alpha_A^\varepsilon:=\frac{V_\varepsilon}{\Delta L}>0.
\]
Since \(C_A^\varepsilon=C_A=(L/\Delta)I_2\),
\[
\alpha_A^\varepsilon C_A^\varepsilon
=\frac{V_\varepsilon}{\Delta^2}I_2
=N_A^\varepsilon.
\]
Hence the first compatibility residual is identically zero:
\[
\frac{1}{\alpha_A^\varepsilon}
\left\|N_A^\varepsilon-\alpha_A^\varepsilon
C_A^\varepsilon\right\|_1=0.
\]

For the second condition, choose
\[
\alpha_B^\varepsilon
:=\frac{2V_\varepsilon}{\Delta^2c_0}>0.
\]
Then \(N_B^\varepsilon/\alpha_B^\varepsilon=(c_0/2)I_2\), and
\[
\begin{aligned}
\frac{1}{\alpha_B^\varepsilon}
\left\|N_B^\varepsilon-\alpha_B^\varepsilon C_B^\varepsilon\right\|_1
&=\left\|\frac{N_B^\varepsilon}{\alpha_B^\varepsilon}
-C_B^\varepsilon\right\|_1\\
&=\left\|d_\varepsilon
\begin{pmatrix}0&-1\\-1&0\end{pmatrix}\right\|_1\\
&=2|d_\varepsilon|\longrightarrow0.
\end{aligned}
\]
Thus both compatibility conditions of Theorem \ref{aT2} hold.

The trace-scale compatibility relation also holds exactly:
\[
\alpha_A^\varepsilon\operatorname{Tr}C_A
=\frac{2V_\varepsilon}{\Delta^2}
=\alpha_B^\varepsilon\operatorname{Tr}C_B.
\]
Therefore Theorem \ref{aT2} applies to the example of Section~3, and both
normalized partial traces converge to \(I_2/2\), the reduced state of
each subsystem of the Bell state obtained there.

\section{Integrability of the Macro-Covariance Operator}

We justify the existence of the macro-covariance operator
$\widehat C$ under the integrability condition
\begin{equation}
\label{eq:fourth-integrability}
\int_{\Delta}
\mathbb{E}\!\left[\|X(t)\|^2\,\|Y(t)\|^2\right] \, dt < \infty,
\end{equation}
where $0<|\Delta|<\infty$. We assume that the processes are jointly
measurable and that the integral defining $C_{\Delta}$ is a well-defined
Bochner integral.

Recall that
\begin{equation}
\label{eq:micro-vector-integrability}
C_{\Delta}
=\frac{1}{|\Delta|}\int_{\Delta}
X(t)\otimes Y(t)\,dt
\in H_A\otimes H_B.
\end{equation}
By Jensen's inequality for the squared norm of a Hilbert-space-valued
integral, for almost every $\omega$ we have
\begin{align}
\|C_{\Delta}(\omega)\|^2
&=\left\|
\frac{1}{|\Delta|}\int_{\Delta}
X(t,\omega)\otimes Y(t,\omega)\,dt
\right\|^2
\nonumber\\
&\leq \frac{1}{|\Delta|}\int_{\Delta}
\|X(t,\omega)\otimes Y(t,\omega)\|^2\,dt.
\label{eq:jensen-bound}
\end{align}
Using the tensor-product identity
$\|x\otimes y\|=\|x\|\,\|y\|$, we obtain
\begin{equation}
\label{eq:jensen-bound2}
\|C_{\Delta}(\omega)\|^2
\leq \frac{1}{|\Delta|}\int_{\Delta}
\|X(t,\omega)\|^2\,\|Y(t,\omega)\|^2\,dt.
\end{equation}
Taking expectations and applying Tonelli's theorem to the nonnegative
integrand gives
\begin{align}
\mathbb{E}\|C_{\Delta}\|^2
&\leq \frac{1}{|\Delta|}
\mathbb{E}\!\left[
\int_{\Delta}\|X(t)\|^2\,\|Y(t)\|^2\,dt
\right]
\nonumber\\
&=\frac{1}{|\Delta|}\int_{\Delta}
\mathbb{E}\!\left[\|X(t)\|^2\,\|Y(t)\|^2\right] \, dt
<\infty.
\label{eq:finite-second-moment}
\end{align}
Thus $C_{\Delta}$ has a finite second moment, and its expectation
$\mathbb{E}[C_{\Delta}]$ exists. The centered random vector
\begin{equation}
Z_{\Delta}=C_{\Delta}-\mathbb{E}[C_{\Delta}]
\end{equation}
also has a finite second moment, since
\begin{equation}
\mathbb{E}\|Z_{\Delta}\|^2
=\mathbb{E}\|C_{\Delta}\|^2
-\|\mathbb{E}[C_{\Delta}]\|^2
<\infty.
\end{equation}

The macro-covariance operator
\begin{equation}
\widehat C
=\mathbb{E}\!\left[
|Z_{\Delta}\rangle\langle Z_{\Delta}|
\right]
\end{equation}
is therefore a positive trace-class operator on $H_A\otimes H_B$, with
\begin{equation}
\label{eq:trace-covariance}
\operatorname{Tr}\widehat C
=\mathbb{E}\|Z_{\Delta}\|^2<\infty.
\end{equation}
This establishes that the macro-covariance operator is well-defined and
has finite trace under condition~\eqref{eq:fourth-integrability}.

\end{myappendices}

\subsubsection*{Author Contributions} 
AK: Conceptualization, development of the general idea and the subinterval allocation scheme, and preparation of the first version of the manuscript. IB: Investigation of the general framework for the partial-trace formula, development of Theorems 1 and 2, and construction of the corresponding examples. Both authors contributed to the final version of the manuscript.


\begin{thebibliography}{99}

\bibitem{ARX} Khrennikov, A. (2026). The Double Covariance Model: A Stochastic Reconstruction of Quantum Entangled States via Interplay of Micro-Macro Time Scales. https://doi.org/10.48550/arXiv.2601.17070 .

\bibitem{Wig32} Wigner, E. (1932). On the quantum correction for thermodynamic equilibrium. \textit{Physical Review}, 40(5), 749.

\bibitem{Hil84} Hillery, M., O'Connell, R. F., Scully, M. O., \& Wigner, E. P. (1984). Distribution functions in physics: Fundamentals. \textit{Physics Reports}, 106(3), 121-167.

\bibitem{DeB27} de Broglie, L. (1927). La m'{e}canique ondulatoire et la structure atomique de la mati`{e}re et du rayonnement. \textit{Journal de Physique et le Radium}, 8(5), 225-241.

\bibitem{DeB60} de Broglie, L. (1960). \textit{Non-linear Wave Mechanics: A Causal Interpretation}. Elsevier.

\bibitem{VN} von Neumann, J. (1932). \textit{Mathematische Grundlagen der Quantenmechanik}. Springer, Berlin.

\bibitem{Bel64} Bell, J. S. (1964). On the Einstein Podolsky Rosen paradox. \textit{Physics Physique Fizika}, 1(3), 195.

\bibitem{Asp82} Aspect, A., Dalibard, P., \& Roger, G. (1982). Experimental test of Bell's inequalities using time-varying analyzers. \textit{Physical Review Letters}, 49(25), 1804.

\bibitem{Bel66} Bell, J. S. (1966). On the problem of hidden variables in quantum mechanics. \textit{Reviews of Modern Physics}, 38(3), 447.

\bibitem{AM1} de la Pena, L., \& Cetto, A. M. (1996). \textit{The Quantum Dice: An Introduction to Stochastic Electrodynamics}. Kluwer Academic Publishers.

\bibitem{AM2} De la Pena, L., Cetto, A. M., \& Valdes-Hernandez, A. (2015). The emerging quantum. The Physics Behind Quantum Mechanics. Cham: Springer International Publishing.

\bibitem{Col77} Cole, D. C. (1977). Transition from classical to quantum mechanics in the framework of stochastic electrodynamics. \textit{Phys. Rev. D}, 15(4), 988.

\bibitem{AM3} Cetto, A.M. and De la Pena, L.,
\emph{Quantum Mechanics: A Physical Approach}, Cambridge University Press, 2025.

\bibitem{Mad27} Madelung, E. (1927). Quantentheorie in hydrodynamischer Form. \textit{Zeitschrift f"{u}r Physik}, 40(3-4), 322-326.

\bibitem{Wya05} Wyatt, R. E. (2005). \textit{Quantum Dynamics with Trajectories: Introduction to Quantum Hydrodynamics}. Springer.

\bibitem{Boh52} Bohm, D. (1952). A suggested interpretation of the quantum theory in terms of ``hidden'' variables. I \& II. \textit{Physical Review}, 85(2), 166.

\bibitem{Dur09} D"{u}rr, D., \& Teufel, S. (2009). \textit{Bohmian Mechanics: The Physics and Mathematics of Quantum Theory}. Springer.

\bibitem{KS67} Kochen, S., \& Specker, E. P. (1967). The problem of hidden variables in quantum mechanics. \textit{Journal of Mathematics and Mechanics}, 59-87.

\bibitem{Mer90} Mermin, N. D. (1990). Simple unified form for the major no-go theorems. \textit{Physical Review Letters}, 65(27), 3373.

\bibitem{Bush1} J.~W.~M. Bush and A.~U. Oza, Hydrodynamic quantum analogs,
Rep. Prog. Phys. \textbf{84}(1), 017001 (2021).

\bibitem{Bush2} Papatryfonos, K., Vervoort, L., Nachbin, A., Labousse, M. and Bush, J.W.M. (2024). Static Bell test in pilot-wave hydrodynamics, Phys. Rev. Fluids \textbf{9}, 084001.

\bibitem{Bush3} Primkulov, B. K., Evans, D. J., Been, J. B., \& Bush, J. W. (2025). Nonresonant effects in pilot-wave hydrodynamics. Physical Review Fluids, 10(1), 013601.

\bibitem{Elze1} Elze, H.T., \emph{Linear dynamics of quantum–classical hybrids},
Phys. Rev. A \textbf{85}, 052109 (2012).

\bibitem{Elze2} H.-T.~Elze, Quantum-classical hybrid dynamics – a summary, \emph{J. Phys.: Conf. Ser.} \textbf{442}, 012007 (2013).

\bibitem{PCSFT1} Khrennikov, A. (2005). A pre-quantum classical statistical model with infinite-dimensional phase space. Journal of Physics A: Mathematical and General, 38(41), 9051.

\bibitem{PCSFT2} Khrennikov, A. (2010). Correlations of components of prequantum field corresponding to biparticle quantum system. Europhysics letters, 90(4), 40004.

\bibitem{PCSFT3} Khrennikov, A. (2010). Description of composite quantum systems by means of classical random fields. Foundations of physics, 40(8), 1051-1064.

\bibitem{Beyond} Khrennikov, A. (2014). Beyond quantum. Singapore: Pan Stanford Publ.

\bibitem{ArxL} A. Khrennikov, I, Basieva, Quantum Markovian Dynamics from a Double Covariance Stochastic Framework.
https://arxiv.org/abs/2605.29508

\bibitem{LA1} Langevin, P. (1908). Sur la theorie de la mouvement brownien. \textit{C. R. Acad. Sci. Paris}, 146, 530-533.

\bibitem{LA2} Nelson, E. (1966). Derivation of the Schrödinger equation from Newtonian mechanics. \textit{Physical Review}, 150(4), 1079.

\bibitem{Haken1975} Haken, H. (1975). Cooperative phenomena in systems far from thermal equilibrium. \textit{Rev. Mod. Phys.}, 47(1), 67.

\bibitem{Haken1983} Haken, H. (1983). \textit{Synergetics: An Introduction}. Springer-Verlag.

\bibitem{Kirkwood} Kirkwood, J. G. (1946). The statistical mechanical theory of transport processes I. General theory. \textit{The Journal of Chemical Physics}, 14(3), 180-201.

\bibitem{Grad1949} Grad, H. (1949). On the kinetic theory of rarefied gases. \textit{Comm. Pure Appl. Math.}, 2(4), 331.

\bibitem{ALA}
A.~E.~Allahverdyan, A.~Khrennikov, and Th.~M.~Nieuwenhuizen, Brownian entanglement,
Phys. Rev. A \textbf{72}, 03210 (2005).
    
\bibitem{R1}
C.~Rovelli,
\emph{Relational Quantum Mechanics},
Int.\ J.\ Theor.\ Phys.\ \textbf{35}, 1637--1678 (1996).

\bibitem{R2} M.~Di~Biagio and C.~Rovelli,
\emph{Relational Quantum Mechanics},
in \emph{The Stanford Encyclopedia of Philosophy},
ed. E.~N.~Zalta,
Fall 2022 Edition.

\bibitem{R3}
E.~Adlam and C.~Rovelli,
\emph{Information Is Physical: Cross-Perspective Links in Relational Quantum Mechanics},
arXiv:2203.13342 [quant-ph] (2022).

\bibitem{ML}
L.~Huynh, J.~Hong, A.~Mian, H.~Suzuki, Y.~Wu, and S.~Camtepe, 
Quantum-Inspired Machine Learning: a Survey, \textit{arXiv preprint arXiv:2308.11269}, 2023.

\bibitem{QL1} A. Khrennikov, Ubiquitous quantum structure: from psychology to finances, Springer, Berlin-Heidelberg-New York, 2010.

\bibitem{QL2} J.R. Busemeyer,  P.D. Bruza, Quantum models of cognition and decision, Cambridge University Press, Cambridge, UK, 2012.

\bibitem{QL3}  F. Bagarello, Quantum Concepts in the Social, Ecological and Biological Sciences, Cambridge University Press, Cambridge, UK, 2019.

\bibitem{K} A. N. Kolmogoroff. Grundbegriffe der Wahrscheinlichkeitsrechnung. Springer, Berlin (1933); Kolmogorov, A. N.: Foundations of the Theory of Probability. Chelsea Publ. Company, New York (1956).

\bibitem{Khr10} Khrennikov, A., Yamada,M. (2025). Quantum-like representation of neuronal networks' activity: modeling ``mental entanglement''. Frontiers in Human Neuroscience, 19, Article 1685339. 

\end{thebibliography}
\end{document}